\documentclass[letterpaper,11pt]{article}

\usepackage[T1]{fontenc}
\usepackage{textcomp}
\usepackage{anyfontsize}
\usepackage{cite}

\usepackage{listings}
\lstdefinelanguage{program}{
  keywords={
    atomic,
    assume,assert,call,return,new,until,
    false,true,duplicate,restart,lock,unlock,
    locate,insert,delete,contains,removeRight,rotateRightLeft,
    rcu_read_lock,rcu_read_unlock,synchronize_rcu
  },  
  morecomment=[l]{//},
  morecomment=[s]{/*}{*/},
  morecomment=[n]{(**}{**)},
  mathescape=true,
  escapeinside=`',
}

\usepackage[left=1in, right=1in, top=1in,bottom=1in]{geometry}

\usepackage{booktabs}
\usepackage{mathpartir}
\usepackage{multirow}
\usepackage{amsmath}
\usepackage{amsthm}
\usepackage{amssymb}
\usepackage{MnSymbol}
\usepackage{stmaryrd}
\usepackage{tabularx}
\usepackage{wrapfig}
\usepackage{array}
\usepackage{subcaption}
\usepackage[T1]{fontenc}
\usepackage[nodayofweek]{datetime}
\usepackage{graphicx}
\usepackage{paralist}
\usepackage[ligature,reserved,inference]{semantic}

\usepackage{xcolor}
\usepackage{extarrows}
\usepackage{mathtools}
\usepackage{paralist}

\usepackage{array}
\usepackage{times}

\lstdefinestyle{lnumbers}{numbers=left,  stepnumber=1, numberblanklines=false, numberstyle=\tiny,basicstyle=\scriptsize, numbersep=3pt, escapeinside={/*}{*/}}
\lstdefinestyle{nonumbers}{numbers=none, escapeinside={/*}{*/}}

\usepackage{cleveref}

\newif\ifnitpick
\nitpickfalse

\newif\iftodobom
\todobomtrue

\newif\iflong
\longtrue

\Crefname{conjecture}{Conjecture}{Conjectures}
\Crefname{proposition}{Proposition}{Propositions}
\Crefname{lemma}{Lemma}{Lemmas}
\Crefname{corollary}{Corollary}{Corollaries}
\Crefname{example}{Example}{Examples}
\Crefname{definition}{Definition}{Definitions}
\Crefname{figure}{Fig.}{Fig.}
\Crefname{section}{Sec.}{Sec.}

\newtheorem{theo}{Theo}[section] 

\newtheorem{definition}[theo]{Definition}
\newtheorem{theorem}[theo]{Theorem}
\newtheorem{lemma}[theo]{Lemma}

\newcommand{\record}[1]{\mathit{record}({#1})}
\newcommand{\recordnoarg}{\mathit{record}}
\newcommand{\guess}[1]{\mathit{guess}({#1})}

\newcommand{\semanticdomainstyle}[1]{
  \ensuremath{\mathchoice
    {\mbox{\normalfont\ensuremath{#1}}}
    {\mbox{\normalfont\ensuremath{#1}}}
    {\mbox{\normalfont\scriptsize\ensuremath{#1}}}
    {\mbox{\normalfont\tiny\ensuremath{#1}}}}}

\reservestyle{\semanticdomain}{\semanticdomainstyle}

\semanticdomain{
  Methods[\mathbb{M}],
  Vals[\mathbb{V}],
  Labels[\mathbb{L}],
  Ops[\mathbb{I}],
  Acts[\mathbb{A}],
}

\newcommand{\proglangkeywordstyle}[1]{\ensuremath{\mathbf{#1}}\xspace}
\reservestyle{\proglangkeyword}{\proglangkeywordstyle}

\lstdefinelanguage{program}{
  basicstyle=\tt
  classoffset=0,
  keywords={var,const,reg,proc,skip,assume,if,then,else,
    while,do,call,return,post,await,ewait,yield,init,
    let,in,and,or,true,false,
    for,from,to,
    future, touch,
    fork, rfork, join,
    async, finish,returns,
    spawn, sync, inlet,until,procedure,
    eventloop,
    foreach,
    atomic, method, wait, signal,
    thread, client, begin, end, spawn, repeat, times,
  },
  keywordstyle=\bf,
  classoffset=1,
  morekeywords={g,l},
  keywordstyle=\tt,
  classoffset=0,
  basicstyle=\tt,
  commentstyle=\itshape,
  morecomment=[l]{//},
  morecomment=[s]{/*}{*/},
  morecomment=[n]{(*}{*)},
  mathescape=true,
  escapeinside=`'
}

\lstnewenvironment{program}{
  
  \lstset{
    language={program},
    columns=flexible,
    breaklines=true,
    tabsize=4,
  }
}{}

\lstMakeShortInline"

\begin{document}

\thispagestyle{empty}
\begin{titlepage}

\title{Putting Strong Linearizability in Context:\\
Preserving Hyperproperties in\\Programs that Use Concurrent Objects}

 \author{Hagit Attiya$^1$ \and Constantin Enea$^2$}

 \date{
 $^1$Technion - Israel Institute of Technology, Israel \\
 \texttt{hagit@cs.technion.ac.il} \\
 $^2$Université de Paris, IRIF, CNRS, Paris, France \\
 \texttt{cenea@irif.fr}
 }

\maketitle

\begin{abstract}
It has been observed that linearizability, the prevalent consistency
condition for implementing concurrent objects, does not preserve some
probability distributions.
A stronger condition, called \emph{strong linearizability} has been
proposed, but its study has been somewhat ad-hoc.
This paper investigates strong linearizability by casting it in the
context of \emph{observational refinement} of objects.
We present a strengthening of observational refinement,
which generalizes strong linearizability,
obtaining several important implications.

When a concrete concurrent object \emph{refines} another,
more abstract object---often sequential---the correctness of a program
employing the concrete object can be verified by
considering its behaviors when using the more abstract object.
This means that \emph{trace properties} of a program using the concrete
object can be proved by considering the program with the abstract object.
This, however, does not hold for \emph{hyperproperties}, including
many security properties and probability distributions of events.

We define \emph{strong observational refinement},
a strengthening of refinement that preserves hyperproperties,
and prove that it is \emph{equivalent}
to the existence of \emph{forward simulations}.
We show that strong observational refinement generalizes
\emph{strong linearizability}.
This implies that strong linearizability is also equivalent to forward
simulation, and shows that strongly linearizable implementations can be
composed both horizontally (i.e., \emph{locality})
and vertically (i.e., with \emph{instantiation}).

For situations where strongly linearizable implementations do not exist
(or are less efficient), we argue that reasoning about hyperproperties
of programs can be simplified by strong observational refinement of
non-atomic abstract objects.
\end{abstract}

\end{titlepage}

\section{Introduction}

\emph{Abstraction} is key to the design and verification of large,
complicated software.
In concurrent programs, featuring intricate interactions between multiple
threads, abstraction is often used to encapsulate low-level shared memory
accesses within high-level abstract data types, called \emph{concurrent objects}.
Arguing about properties of such a program $P$ is greatly simplified
by considering a concurrent object as a \emph{refinement} of another,
more abstract one:
a \emph{concrete} object $O_1$ is said to \emph{observationally refine}
another, \emph{abstract} object $O_2$ if any
behavior $P$ can observe with $O_1$ is also observed by $P$ with $O_2$.
When $O_2$ is an \emph{atomic object},
in which each operation is applied in exclusion, observational refinement
is equivalent to linearizability~\cite{DBLP:conf/popl/BouajjaniEEH15,
DBLP:journals/tcs/FilipovicORY10}.\footnote{
    \emph{Linearizability}~\cite{DBLP:journals/toplas/HerlihyW90}
    states that a concurrent execution of operations corresponds
    to some serial sequence of the same operations permitted by the specification.}

Intuitively, linearizability, and more generally, observational refinement,
seem to imply that anything we can prove
about $P$ with $O_2$ also holds when $P$ executes with $O_1$.
This is indeed the case when considering \emph{trace properties},
i.e., properties that are specified as \emph{sets of traces},
in particular, safety properties.

Unfortunately, many interesting properties cannot be specified as
properties of individual traces, i.e., as trace properties.
Notable examples are security properties such as
\emph{noninterference}~\cite{DBLP:conf/sp/GoguenM82a},
stipulating that commands executed by users with high clearance have
no effect on system behavior observed by users with low clearance.
Other examples are quantitative properties like bounds on 
the \emph{probability distribution} of events, e.g., 
the mean response time over sets of executions.
Indeed, while the fact that \emph{the average response time of an
operation in an execution is smaller than some bound $X$} is a trace property,
the requirement that the \emph{average response time over all 
executions is smaller than $X$} cannot be stated as a trace property.

\emph{Hyperproperties}~\cite{DBLP:journals/jcs/ClarksonS10}, namely,
sets of sets of traces, allow to capture such expectations.
By definition, every property of system behavior (for systems modeled
as trace sets) can be specified as a hyperproperty.
It is known that observational refinement does not preserve
hyperproperties~\cite{DBLP:conf/sp/McLean94}, in general.
More recently,
it has been shown that linearizability does not preserve probability
distributions over traces~\cite{DBLP:conf/stoc/GolabHW11},
allowing an adversary scheduler additional control over
the possible outcomes of a distributed randomized program.
(An example appears in Section~\ref{sec:motivation}.)

This paper defines the notion of \emph{strong observational refinement},
relates it to hyperproperties, and shows its equivalence to
\emph{forward simulations}.
We show that strong observational refinement generalizes strong
linearizability~\cite{DBLP:conf/stoc/GolabHW11}.\footnote{
    \emph{Strong linearizability} requires that \emph{the
    linearization of a prefix of a concurrent execution is a prefix
    of the linearization of the whole execution},
    see Section~\ref{section:forward simulations}.}
We also explore the possibility of using---instead of the classical sequential
specifications---\emph{concurrent} specifications, which are nevertheless simpler.

To explain our results in more detail, consider a \emph{labeled
transition system} (\emph{LTS}) that, intuitively,
represents all the executions of the object under the most general client
(that may call methods in any order and from any thread).
A state of the LTS corresponds to a state of the object and
transitions correspond to method calls/returns,
or internal steps within a method invocation.
A \emph{sequential specification} corresponds to a concurrent object
where essentially, method bodies consist of a single atomic step
that acts according to the sequential specification
(hence,
they are totally ordered in time during any execution).

An LTS $O_1$ \emph{observationally refines} an LTS $O_2$ if and only if
the \emph{histories} (i.e., sequences of call/return actions)
generated by $O_1$ are included in
those generated by $O_2$~\cite{DBLP:conf/popl/BouajjaniEEH15}.
In this way, observational refinement of two LTSs reduces to a inclusion
between their traces, when projected over some alphabet $\Gamma$
(in this case, $\Gamma$ is the set of call/return actions),
called \emph{$\Gamma$-refinement}.

A \emph{forward simulation} maps every step of $O_1$ to a sequence of steps of $O_2$,
starting from the initial state of $O_1$ and advancing in a forward manner;
a \emph{backward simulation} is similar, but it goes in the reverse direction,
from end states back to initial states.
When proving linearizability, an important special case of forward
simulation is the identification of \emph{fixed} linearization points.
A forward/backward simulation can be parameterized by an alphabet $\Gamma$,
in which case the sequence of steps of $O_2$ associated to a step of $O_1$
should contain a step labeled by an action $a\in \Gamma$ if and only if
the step of $O_1$ is also labeled by $a$.
It is known~\cite{DBLP:journals/iandc/LynchV95} that $\Gamma$-refinement
holds if and only if there is a combination of $\Gamma$-forward and $\Gamma$-backward
simulations from $O_1$ to $O_2$; a forward simulation suffices when
the projection of $O_2$ on $\Gamma$ is deterministic.
(See Section~\ref{sec:lts}.) 

The notion of strong observational refinement relies on the concept
of a \emph{deterministic scheduler} that resolves the non-determinism
introduced for instance, from
the execution of internal actions by parallel threads
(it is similar to the notion of strong adversary introduced in the
context of randomized algorithms~\cite{DBLP:journals/dc/Aspnes03}).
$O_1$ \emph{strongly (observationally) refines} $O_2$ if a
program $P$ running under a deterministic schedule with $O_1$ makes
the same observations as when $P$ runs with $O_2$
with a possibly-different deterministic schedule.
(The complete definition appears in Section~\ref{sec:sor}.)
We prove that strong observational refinement implies the existence of a forward simulation.
The converse direction is fairly straightforward, proving the equivalence
of these two notions, and imply compositional proof methodologies. 
(These results appear in Section~\ref{section:forward simulations}.)
By relating strong linearizability to strong observational refinement,
we prove that a concrete object is a strong linearization of an atomic
object \emph{if and only if}
there is an appropriate forward simulation between the two
(Appendix~\ref{app:strong-lin}).
This immediately implies methods for composing strongly linearizable
concurrent objects.

To address situations where there is no concrete object that
strongly linearizes a particular atomic object
\cite{DBLP:conf/podc/HelmiHW12,DBLP:conf/wdag/DenysyukW15},
or in cases where such an object is less efficient,
we suggest concrete objects that strongly refine other, more abstract
objects that still expose some concurrency.
This follows~\cite{DBLP:conf/cav/BouajjaniEEM17} and allows to
simplify reasoning about randomized programs even when using objects like
the Herlihy\&Wing queue~\cite{DBLP:journals/toplas/HerlihyW90}
or snapshot objects~\cite{DBLP:journals/jacm/AfekADGMS93},
which are not strongly linearizable.
For example,
in the case of atomic snapshots,
the abstract object that obtains several instantaneous snapshots during
the \emph{scan} operation and then \emph{arbitrarily} returns one of them
(see Section~\ref{sec:applications}).
Arguing about a program using this abstract object is simpler, while
still exposing the power of an adversarial scheduler to manipulate
the responses of a scan.

 \section{Motivating Example: A Stack Implementation that Leaks Information}
\label{sec:motivation}

When an object $O_1$ refines a specification object $O_2$,
any safety property of a program $P$ using $O_2$
(that refers only to $P$'s state and is agnostic to the internal state of the object $O_2$)
is preserved when $O_2$ is replaced by its refinement $O_1$.
However, refinement does not preserve \emph{hyperproperties}~\cite{DBLP:journals/jcs/ClarksonS10},
i.e., properties of \emph{sets} of traces.

\begin{figure}[bt]
\center
  \footnotesize
  \begin{minipage}[c]{28mm}
    \begin{program}
a = push(0);
low1 = pop();
    \end{program}
  \end{minipage}
  \begin{minipage}[c]{5mm}
  $||$
  \end{minipage}
  \begin{minipage}[c]{28mm}
    \begin{program}
b = push(1);
low2 = pop();
    \end{program}
  \end{minipage}
  \begin{minipage}[c]{5mm}
  $||$
  \end{minipage}
  \begin{minipage}[c]{50mm}
    \begin{program}
assume a == b == OK;
push(2);
high = highBooleanInput();
    \end{program}
  \end{minipage}
\vspace{-3mm}
  \footnotesize \caption{ A program with three threads using a concurrent stack (we assume that {\tt push} returns the value {\tt OK}). Statements in the same thread are aligned vertically. The statement {\tt assume} blocks the program when the Boolean condition is not satisfied and {\tt highBooleanInput} returns a Boolean value labeled as high clearance. The {\tt assume} statement enforces that {\tt push(0)} and {\tt push(1)} finish before {\tt push(2)} starts.}
  \label{fig:prog}
\vspace{-4mm}
\end{figure}

We explain this issue by considering \emph{noninterference}~\cite{DBLP:conf/sp/GoguenM82a}
in the program of Figure~\ref{fig:prog}.
This program invokes methods of a concurrent stack, and we wish to show that
independently of the thread scheduler,
none of the low clearance variables {\tt low1} and {\tt low2} can leak the value of the high clearance variable ${\tt high}$, i.e., it is impossible to define a thread scheduler which admits only executions where ${\tt low1}={\tt high}$ or only executions where ${\tt low2}={\tt high}$. A precise notion of scheduler will be defined below, but for now, it is enough to think of a thread scheduler as a monitor that chooses to activate threads depending on the history of the execution. This property is satisfied by the program when invoking an \emph{atomic} concurrent stack.
Indeed, assuming that {\tt push(0)} is scheduled before the one of {\tt push(1)} 
(the other case is similar), then either 
(i) {\tt push(2)} is scheduled before at least one of the {\tt pop} invocations, 
and then, ${\tt low1},{\tt low2}\in\{1,2\}$ which shows that none of these 
two variables equals the value of ${\tt high}$ when ${\tt high}=0$, or 
(ii) {\tt push(2)} is scheduled after the {\tt pop} invocations, and then, 
the scheduler admits executions where ${\tt low1}=b_1$, ${\tt low2}=b_2$, 
and ${\tt high}=0$ if and only if it admits executions where 
${\tt low1}=b_1$, ${\tt low2}=b_2$, and ${\tt high}=1$ (for $b_1, b_2\in\{0,1,2\}$).

\begin{wrapfigure}{l}{7.8cm}
\includegraphics[scale=.59]{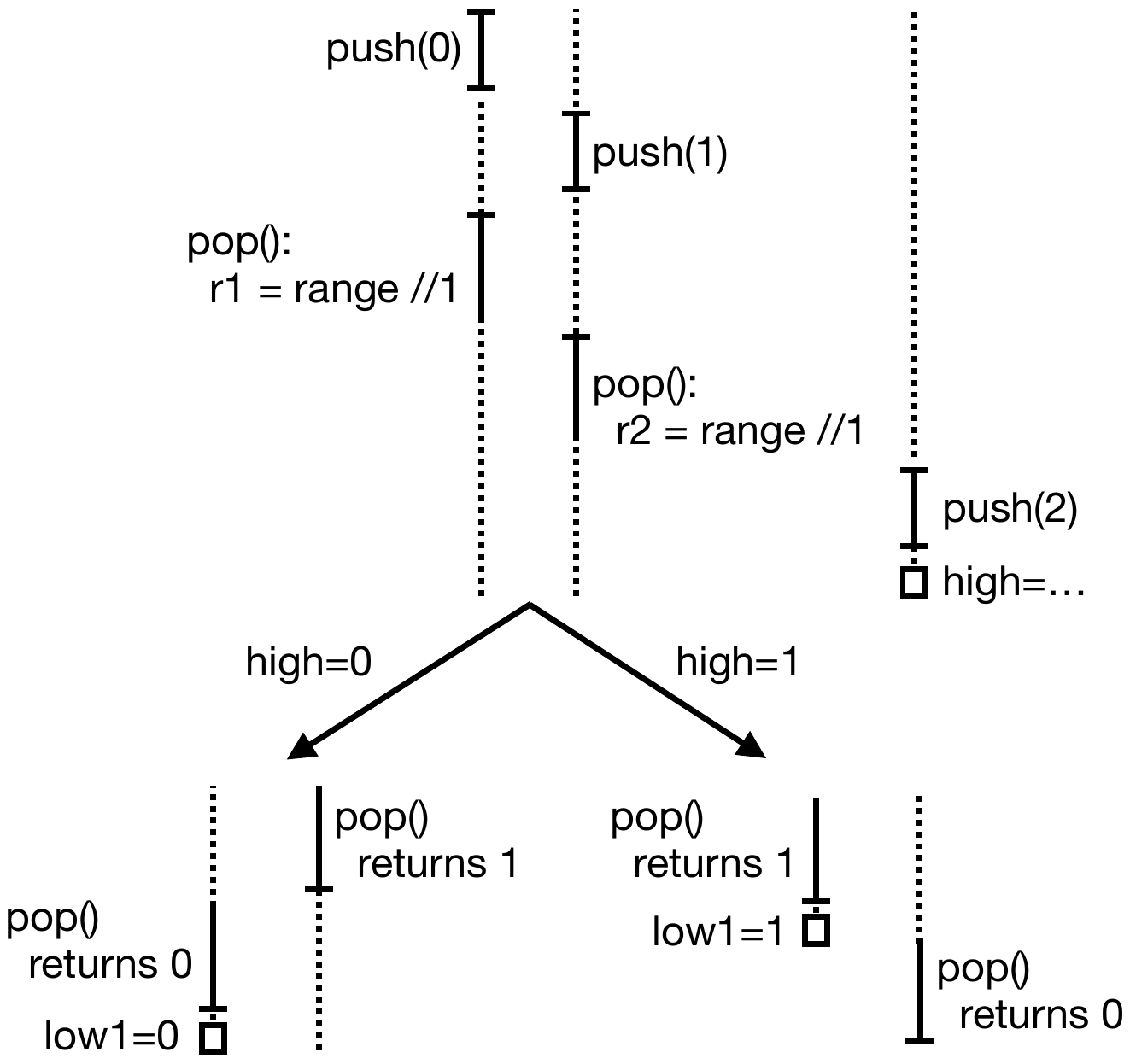}
  \footnotesize \caption{ A scheduler for the program in Figure ~\ref{fig:prog}. Time flows from top to bottom. Dotted edges denote periods of time where a thread is not active.}
  \label{fig:sched}
\vspace{-2mm}
\end{wrapfigure}
This property is however not satisfied by this program when using the 
concurrent stack of Afek et al.~\cite{DBLP:journals/dc/AfekGM07}.
This stack stores the elements into an infinite array {\tt items};
a shared variable {\tt range} keeps the index of the first unused position
in {\tt items}.
The push method stores the input value in the array while also incrementing
{\tt range} (the details are irrelevant for our example).
The pop method first reads {\tt range} and then traverses the array backwards
starting from this position, until it finds a position storing a non-null
element (array cells can be nullified by concurrent pop invocations).
It atomically reads this element and stores null in its place.
If the pop reaches the bottom of the array without finding non-null cells,
then it returns that the stack is empty.
Unlike the case of atomic stacks, Figure~\ref{fig:sched} shows
a thread scheduler where {\tt low1} stores the value of {\tt high}.
This scheduler imposes that {\tt push(0)} executes before {\tt push(1)}
(so that 0 occurs before 1 in the array {\tt items}),\footnote{
A similar scheduler can be defined when {\tt push(1)} executes before {\tt push(0)}.}
and then preempts pop invocations just after reading the value of {\tt range}
which equals 1 (assuming the array indexing starts at 0).
Then, it schedules the third thread and, depending on the value of {\tt high},
it schedules the rest of the {\tt pop} invocations such that the {\tt pop}
in the first thread extracts a value which equals {\tt high}.
This ensures that {\tt low1} == {\tt high}.

This shows that noninterference in programs invoking the atomic stack 
is not preserved when the latter is replaced by  the concurrent stack of
Afek et al.~\cite{DBLP:journals/dc/AfekGM07}, 
although the latter is a refinement of the atomic stack.
Section~\ref{sec:sor} presents a stronger notion of observational refinement 
that preserves hyperproperties and in particular, noninterference.
 \section{Modelling Concurrent Objects as Labeled Transition Systems}
\label{sec:lts}
\label{sec:objects}

\emph{Labeled transition systems} (\emph{LTS}) capture
shared-memory programs with an arbitrary number of threads,
abstracting away the details of any particular programming system
irrelevant to our development.

An LTS $A=(Q,\Sigma,  s_0, \delta)$ over the possibly-infinite alphabet
$\Sigma$ is a possibly-infinite set $Q$ of states with
initial state $s_0 \in Q$, and a transition relation $\delta \subseteq Q \times \Sigma \times
Q$.
The $i$th symbol of a sequence $\tau \in \Sigma^*$ is denoted $\tau_i$, and $\epsilon$ is the empty
sequence.
An \emph{execution} of $A$ is an alternating sequence of states and transition labels (also called \emph{actions})
$\rho = s_0, a_0,s_1\ldots a_{k-1},s_k$ for some $k>0$ such that $(s_i, a_i, s_{i+1})\in \delta$
for each $0\leq i<k$. We write $s_i\xrightarrow{a_i\ldots a_{j-1}}_A s_j$ as shorthand for
the subsequence $s_i,a_i,...,s_{j-1},a_{j-1},s_j$ of $\rho$. 
(in particular $s_i\xrightarrow{\epsilon}s_i$).

The projection $\tau| \Gamma$ of a sequence $\tau$ is the maximum subsequence of $\tau$ over
 alphabet $\Gamma$. This notation is extended to sets of sequences as usual.
A \emph{trace} of $A$ is the projection $\rho | \Sigma$ of an execution $\rho$ of $A$.
The set of executions, resp., traces, of an LTS $A$ is denoted by $\mathit{E}(A)$, resp., $\mathit{T}(A)$.
An LTS is \emph{deterministic} if for any state $s$ and any sequence $\tau\in \Sigma^*$, there is at most
one state $s'$ such that $s\xrightarrow{\tau}s'$. More generally, for an alphabet $\Gamma\subseteq \Sigma$,
an LTS is \emph{$\Gamma$-deterministic} if for any state s and any sequence $\tau\in \Gamma^*$, there
is at most one state $s'$ such that $s\xrightarrow{\tau'}s'$ and $\tau$ is a subsequence of $\tau'$.

An \emph{object} is a \emph{deterministic} LTS over alphabet $C \cup R\cup \Sigma_o$ where $C$, resp., $R$, is the set of call, resp., return, actions, and $\Sigma_o$ is an alphabet of internal actions. Formally, a call action $call(m,d,k)$, resp., a return action $ret(m,d,k)$, combines a method $m$ and argument, resp., return value, $d$ with an operation identifier $k$. Operation identifiers are used to pair call and return actions.
We assume that the traces of an object satisfy standard well-formedness properties,
e.g., return actions correspond to previous call actions.
Given a standard description of an object implementation as a set of methods, its LTS represents the executions of its most general client (that may call methods in any order and from any thread). The states of the LTS represent the shared state of the object together with the local state of each thread. The transitions correspond to statements in the method bodies (in which case they are labeled by internal actions in $\Sigma_o$), or call and return actions. For simplicity, we ignore the association of method invocations to threads since it is irrelevant to our development.
A trace $\tau$ of an object $O$ projected over call and return actions is called a \emph{history} of $O$, and it is denoted by $\mathit{hist}(\tau)$.
The set of histories admitted by an object $O$ is denoted by $H(O)$.
Call and return actions $call(m,\_,k)$ and $ret(m,\_,k)$ are called \emph{matching} when they contain the same operation identifier. A call action is called \emph{unmatched} in a history $h$ when $h$ does not contain the matching return.
A history $h$ is called \emph{sequential} if every call $call(m,d,k)$ is immediately followed by the matching return $ret(m,\_,k)$. Otherwise, it is called \emph{concurrent}.

\emph{Linearizability}~\cite{DBLP:journals/toplas/HerlihyW90} is a standard correctness criterion for concurrent objects expressing conformance to a given sequential specification. This criterion is based on a relation $\sqsubseteq$ between histories: $h_1\sqsubseteq h_2$ iff there exists a well-formed execution $h_1'$ obtained from $h_1$ by appending return actions that correspond to unmatched call actions in $h_1$ or deleting unmatched call actions, such that $h_2$ is a permutation of $h_1'$ that preserves the order between return and call actions, i.e., a given return action occurs before a given call action in $h_1'$ iff the same holds in $h_2$. We say that $h_2$ is a \emph{linearization} of $h_1$.
A history $h_1$ is called \emph{linearizable} w.r.t. an object $O_2$ iff there exists a sequential history $h_2\in H(O_2)$ such that $h_1\sqsubseteq h_2$. An object $O_1$ is linearizable w.r.t. $O_2$, written $O_1\sqsubseteq O_2$, iff each history $h_1\in H(O_1)$ is linearizable w.r.t. $O_2$.

Linearizability has been shown equivalent to a criterion called \emph{observational refinement} which states that every behavior of every program possible using a concrete object would also be possible were the abstract object used instead~\cite{DBLP:conf/popl/BouajjaniEEH15,DBLP:journals/tcs/FilipovicORY10} (the precise meaning of behavior is given below). Actually, this result holds only when the abstract object is \emph{atomic}, i.e., it admits every history which is linearizable w.r.t. a set of \emph{sequential} histories $Seq$ (formally, $H(O)=\{h: \exists h'\in\mathit{Seq}.\ h\sqsubseteq h'\}$). Intuitively, an atomic object corresponds to an implementation where the methods of a sequential object are guarded by global-lock acquisition.

A \emph{program} is a \emph{deterministic} LTS over alphabet $C \cup R\cup \Sigma_p$ where $\Sigma_p$ is an alphabet of \emph{program actions}. Program actions can be interpreted for instance, as assignments to some program variables which are disjoint from the variables used by the object, or as different outcomes of a random choice. Call and return actions represent the interaction between the program and the object. The executions of a program $P$ with an object $O$ are obtained as the executions of the LTS product $P\times O$~\footnote{The \emph{product} $A_1 \times A_2$ of two LTSs is defined as usual, respecting $E(A_1 \times A_2)|(\Sigma_1 \cap \Sigma_2) =
E(A_1)|\Sigma_2 \cap E(A_2)|\Sigma_1$.}. As program and object alphabets only intersect on call and return actions, our formalization supposes that programs and objects communicate only through method calls and returns, and not, e.g.,~through additional shared random-access memory.

Observational refinement between objects $O_1$ and $O_2$ means that any ``observation'' extracted from a program execution possible with $O_1$ (referred to as a ``concrete'' object), is also possible with $O_2$ (referred to as the ``specification''), where an ``observation'' in this context is the projection over the program actions.

\vspace{-2.5mm}
\begin{definition}

  The object $O_1$ \emph{observationally refines} $O_2$, written $O_1 \leq O_2$, iff
  \vspace{-3mm}
  \begin{align*}
    T(P \times O_1)|\Sigma_p \ \subseteq \ T(P \times O_2)|\Sigma_p \\[-8.5mm]
  \end{align*}
  for all programs $P$ over alphabet $\Sigma_p\cup C \cup R$.
\vspace{-2mm}
\end{definition}

The following theorem relates observational refinement to a standard notion 
of refinement between LTSs, defined roughly as inclusion of traces, 
in the context of concurrent objects.\footnote{
This relationship has been shown under natural assumptions about objects and programs~\cite{DBLP:conf/popl/BouajjaniEEH15}. For instance, concerning objects, it is assumed that call actions cannot be disabled and they cannot disable other actions (they can be reordered to the left while preserving the computation), and return actions cannot enable other actions.}
For two LTSs $A_1$ and $A_2$, we say that $A_1$ \emph{refines} $A_2$ when $T(A_1)\subseteq T(A_2)$. More generally, for an alphabet $\Gamma$, $A_1$ \emph{$\Gamma$-refines} $A_2$ when $T(A_1)| \Gamma\subseteq T(A_2)| \Gamma$. By an abuse of notation, $A_1\subseteq_\Gamma A_2$ denotes the fact that $A_1$ $\Gamma$-refines $A_2$ (we will omit $\Gamma$ when it is understood from the context).
Intuitively, the alphabet $\Gamma$ represents a set of actions which are ``observable'' in both $A_1$ and $A_2$, the actions not in $\Gamma$ are considered to be ``internal'' to $A_1$ or $A_2$.
Observational refinement is equivalent to $(C\cup R)$-refinement which means that the histories of the concrete object are included in those of the specification (note that ``plain'' refinement would not hold because the internal actions may differ).

\vspace{-1mm}
\begin{theorem}[\cite{DBLP:conf/popl/BouajjaniEEH15,DBLP:journals/tcs/FilipovicORY10}]\label{th:ref}
$O_1\leq O_2$ iff $O_1\subseteq_{C\cup R} O_2$. If $O_2$ is atomic, then $O_1\leq O_2$ iff $O_1\sqsubseteq O_2$.
\vspace{-2mm}
\end{theorem}

In the rest of the paper, since observational refinement and 
($C\cup R$)-refinement are equivalent, 
we will not make the distinction between the two and refer to both as \emph{refinement}.

 \section{Strong Observational Refinement}
\label{sec:sor}

\vspace{-1mm}
As discussed in Section~\ref{sec:motivation},
refinement does not preserve hyperproperties, which are properties of \emph{sets} of traces and not individual traces as in the case of safety properties.
In the following, we define a stronger notion of observational refinement that preserves such properties, using a notion of scheduler that is actually just a mechanism for resolving the non-determinism induced by internal actions, irrespectively of whether it comes from executing a set of parallel threads.

A \emph{scheduler} for a deterministic LTS $A=(Q,\Sigma,  s_0, \delta)$ over alphabet $\Sigma$ is a function $S:\Sigma^*\rightarrow 2^{\Sigma}$ which prescribes a possible set of
next actions to continue an execution based on a sequence of previous actions. A trace $\tau = a_0\cdot \ldots\cdot a_{k-1}$
is \emph{consistent} with a scheduler $S$ if $a_i\in S(a_0\cdot\ldots\cdot a_{i-1})$ for all $0\leq i<k$ (where by an abuse of notation, $a_0\cdot a_{-1}$ represents the empty sequence).
The set of executions of an LTS $A$ consistent with a scheduler $S$ can be defined using an LTS which is the product between $A$ and an LTS $A_S$ whose states are sequences in $\Sigma^*$ and the transitions link a state $\tau\in\Sigma^*$ to a state $\tau\cdot a\in \Sigma^*$ provided that $S(\tau)=a$ (such a transition is labeled by $a$). Let $T(A,S)$ denote the set of traces of $A$ consistent with $S$.
A scheduler is \emph{admitted} by $A$ if for every $k$, if $\tau = a_0\cdot \ldots\cdot a_{k-1}$ is a trace of $A$ consistent with $S$, then $S(a_0\cdot\ldots\cdot a_{k-1})$ is non-empty and every $a \in S(a_0\cdot\ldots\cdot a_{k-1})$ is enabled in the state $s_k$ with $s_0\xrightarrow{a_0\cdot\ldots\cdot a_{k-1}}_A s_k$.

A scheduler of an LTS $P \times O$ (the product of a program $P$ and an object $O$) is called \emph{deterministic} when it fixes in a unique way the actions of $O$  to continue an execution, i.e., for every sequence $\tau$, $S(\tau)\subseteq \Sigma_p$ or $|S(\tau)|=1$ (where $\Sigma_p$ is the set of program actions). When program actions represent outcomes of random choices made by the program, a deterministic scheduler can be used to model a strong adversary~\cite{DBLP:journals/dc/Aspnes03} which schedules threads depending on those outcomes.
An object $O_1$ strongly (observationally) refines an object $O_2$ if any deterministic schedule admitted by a program $P$ when using $O_1$ leads to exactly the same set of ``observations'' as a deterministic schedule admitted by $P$ were $O_2$ used instead. Formally,

\vspace{-1.5mm}
\begin{definition}\label{def:strong_ref}
  The object $O_1$ \emph{strongly (observationally) refines} $O_2$, written $O_1 \leq_s O_2$, iff
  \vspace{-2.5mm}
  \begin{align*}
    &\mbox{for every deterministic scheduler $S_1$ admitted by $P \times O_1$,}\\[-0.8mm]
    &\hspace{1cm}\mbox{there exists a deterministic scheduler $S_2$ admitted by $P \times O_2$,} \\[-0.8mm]
    &\hspace{2cm}\mbox{such that $T(P \times O_1,S_1)|\Sigma_p=T(P \times O_2,S_2)|\Sigma_p$} \\[-8mm]
  \end{align*}
  for all programs $P$ over alphabet $\Sigma_p\cup C \cup R$.

\vspace{-1.5mm}
\end{definition}

A \emph{hyperproperty} of a program $P$ over alphabet  $\Sigma_p\cup C \cup R$ is a set of sets of sequences over $\Sigma_p$. For instance, the hyperproperty discussed in the context of the program in Figure~\ref{fig:prog} is the set of all sets $T$ s.t.
\vspace{-2mm}
\begin{align*}
(\exists \tau\in T.\ \texttt{low1}(\tau)\neq \texttt{high}(\tau)) \land (\exists \tau\in T.\ \texttt{low2}(\tau)\neq \texttt{high}(\tau))\\[-8mm]
\end{align*}
where for any variable $x$, $x(t)$ is the value of $x$ at the end of trace $t$. A hyperproperty $\varphi$ is \emph{satisfied} by a program $P$ with an object $O$, written $P\times O\models\varphi$, if $T(P \times O,S)|\Sigma_p\in \varphi$ for every deterministic scheduler $S$.

\vspace{-1mm}
\begin{theorem}\label{th:sor_preservation}
If $O_1\leq_s O_2$, then $P\times O_2\models\varphi$ implies $P\times O_1\models\varphi$ for every hyperproperty $\varphi$ of $P$.
\vspace{-2mm}
\end{theorem}

This preservation result applies to \emph{probabilistic} hyperproperties as well, for instance when reasoning about randomized consensus protocols~\cite{DBLP:journals/dc/Aspnes03}. Since a deterministic scheduler fixes in a unique way the object's actions to continue an execution, probability distributions can be assigned only to actions which are internal to the program $P$. This holds for randomized protocols, where randomization is due to coin flip operations that are internal to the protocol and do not concern the behavior of the objects it invokes. Then, the probabilities associated with program actions can be encoded in the action names, thereby encoding probabilistic (hyper)properties as properties of (sets of) traces (see~\cite{DBLP:journals/jcs/ClarksonS10} for more details).

 \section{Characterizing Strong Refinement Using Forward Simulations}
\label{section:forward simulations}

In general, proving refinement between two LTSs relies on \emph{simulation relations} which roughly, are relations between the states of the two LTSs showing that one can mimic every step of the other one. \emph{Forward} simulations show that every outgoing transition from a given state can be mimicked by the other LTS while \emph{backward} simulations show the same for every incoming transition to a given state. Applying induction, forward simulations show that every trace of an LTS is admitted by the other LTS starting from initial states and advancing in a forward manner, while backward simulations consider the backward direction, from end states to initial states.
It has been shown that ($\Gamma$-)refinement is equivalent to the existence of a composition of forward and backward simulations, and to the existence of only a forward simulation provided that $A_2$ is ($\Gamma$-)deterministic~\cite{DBLP:journals/iandc/LynchV95}. In the following, we show that strong observational refinement is equivalent to the existence of a \emph{forward} simulation, which implies that refinement is strictly weaker than strong observational refinement (forward simulations do not suffice to establish refinement in general).

\vspace{-1.5mm}
\begin{definition}
Let $A_1=(Q_1,\Sigma_1,  s_0^1, \delta_1)$ and $A_2=(Q_2,\Sigma_2,  s_0^2, \delta_2)$ be two LTSs and $\Gamma$ an alphabet. A relation $F \subseteq Q_{1} \times Q_{2}$ is called a \emph{$\Gamma$-forward simulation} from $A_1$ to $A_2$ if{f} $(s_0^1,s_0^2)\in F$ and:
\vspace{-2mm}
\begin{itemize}
\item for all $s_1,s_1'\in Q_1$, $a\in \Sigma_1$, and $s_2\in Q_2$, such that $(s_1,a,s_1') \in \delta_1$ and $(s_1,s_2)\in F$, we have that there exists $s_2'\in Q_2$ such that $(s_1',s_2')\in F$ and $s_2 \xrightarrow{\tau}_{A_2} s_2'$ and $\tau| \Gamma=a| \Gamma$. 
\end{itemize}
\vspace{-3.5mm}
\end{definition}

A $\Gamma$-forward simulation states that every step of $A_1$ is simulated by a sequence of steps of $A_2$ (this sequence can be empty to allow for stuttering). Since it should imply that $A_1$ $\Gamma$-refines $A_2$, every step of $A_1$ labeled by an observable action $a\in \Gamma$ should be simulated by a sequence of steps of $A_2$ where exactly one transition is labeled by $a$ and all the other transitions are labeled by non-observable actions (this is implied by $\tau| \Gamma=a| \Gamma$). Also, every internal step of $A_1$ should be simulated by a sequence of internal steps of $A_2$. 

An instantiation of forward simulations are linearizability proofs using the so-called ``fixed linearization points''. Linearizability of a history can be proved by showing that each invocation can be seen as happening at some point, called linearization point, occurring somewhere between the call and return actions of that invocation. Then, the linearization points are \emph{fixed} when they are mapped to a certain fixed set of statements (usually, one statement per method).
This defines a mapping between steps of a concrete implementation and steps of an atomic object, i.e., those fixed statements map to linearization point actions in the atomic object and all the other statements correspond to stuttering steps of the atomic object, thereby defining a forward simulation between the two. As a side remark, backward simulation is necessary to prove linearizability w.r.t.\ atomic specifications, when linearization points depend on future steps in the execution, the Herlihy\&Wing queue~\cite{DBLP:journals/toplas/HerlihyW90} being a classic  example (Schellhorn et al.~\cite{DBLP:conf/cav/SchellhornWD12} present such a proof).

The easier direction is showing that forward simulations imply strong refinement. A forward simulation from $O_1$ to $O_2$ can be used to simulate any scheduler $S_1$ of a program $P$ using $O_1$ by a scheduler of the same program $P$ when using $O_2$. Program actions will be replayed exactly as in $S_1$ while the actions of $O_2$ simulating actions of $O_1$ can be chosen according to the forward simulation (see Appendix~\ref{app:suff}).

\vspace{-1mm}
\begin{lemma}\label{lem:fsim1}
If there exists a $(C\cup R)$-forward simulation from $O_1$ to $O_2$, then $O_1\leq_s O_2$.
\vspace{-1mm}
\end{lemma}

We now prove our key technical result: strong observational refinement (from $O_1$ to $O_2$) implies the existence of a $(C\cup R)$-forward simulation (from $O_1$ to $O_2$). Since the latter implies refinement, a corollary of this result is that strong observational refinement implies observational refinement.
Thus, we define a program $P$ which corresponds to the most general client (of $O_1$) and which uses particular program actions to guess the possible continuations of a given execution with call and return actions. Then, we define a scheduler $S_1$ which ensures that the executions of $P$ with $O_1$ are consistent with the guesses made by the program.  By strong observational refinement, there exists a scheduler $S_2$ such that $P$ produces the same sequences of ``guess'' actions and call/return actions when using $O_2$ and constrained by $S_2$ as when using $O_1$ and constrained by $S_1$ (the preservation of call/return actions is not guaranteed explicitly by strong observational refinement, but it can be enforced using additional program actions used to record them).
If $\Gamma$ is the union of the set of ``guess'' actions and the set of call/return actions, then the program $P$ used in conjunction with $O_2$ and constrained by the scheduler $S_2$ is $\Gamma$-deterministic. Therefore, there exists a forward simulation between the two variations of $P$. Because the program states are disjoint from the object states, this forward simulation between programs leads to a forward simulation between  objects.

\vspace{-1.5mm}
\begin{lemma}\label{lem:fsim1}
If $O_1\leq_s O_2$, then there exists a $(C\cup R)$-forward simulation from $O_1$ to $O_2$.
\vspace{-2mm}
\end{lemma}
\begin{proof}
Let $\Sigma_p=\{ \record{a}, \guess{H}: a\in C\cup R, H\subseteq (C\cup R)^*\}$ be a set of program actions for recording a call/return action $a$ ($\record{a}$) or guessing a set $H$ of possible continuations with sequences of call/return actions ($\guess{H}$). We define a program $P$ with a single state and self-loop transitions labeled by all symbols in $\Sigma_p\cup C\cup R$, i.e., $P=(\{s_0\},\Sigma_p\cup C\cup R,  s_0, \delta)$ where $\delta(s_0,\alpha,s_0)$ for all $\alpha\in \Sigma_p\cup C\cup R$.

We define a deterministic scheduler $S_1$ which ensures that the guesses made by $P$ when using $O_1$ are correct, and that the call/return actions are tracked correctly using $\recordnoarg$ actions.
To ensure the correctness of guesses, we define a mapping $\mathit{after}_1:Q_1\rightarrow 2^{(C\cup R)^*}$ which associates every state $s$ with the set of call/return sequences admitted from $s$, i.e.,
$
\mathit{after}_1(s)=\{\sigma: \sigma\in (C\cup R)^*, \exists s'.\ s\xrightarrow{\tau}_{O_1} s'\land \tau|(C\cup R)=\sigma\}. 
$

Let $S_1$ be a deterministic scheduler such that for every $a_0,\ldots,a_{k-1}\in\Sigma_1$ and $k\geq 0$,
\vspace{-3mm}
\begin{align*}
S_1(a_0\cdot \ldots\cdot a_{k-1}) & =\record{a_{k-1}}\mbox{ if $a_{k-1}\in C\cup R$ and $k\geq 1$} \\[-0.5mm]
S_1(a_0\cdot \ldots\cdot a_{k-1}[\cdot \record{a_{k-1}}]) & =\{ \guess{H}: \exists a.\ s_0^1\xrightarrow{a_0\cdot\ldots\cdot a_{k-1}| \Sigma_1}_{O_1} s\xrightarrow{a}_{O_1} s'
\mbox{ and }H=\mathit{after}_1(s')\}\\[-0.5mm]
&\hspace{5cm}\mbox{ if $a_{k-1}\not\in \Sigma_p$} \\
S_1(a_0\cdot \ldots\cdot a_{k-1}\cdot \guess{H} ) & = a, \mbox{ for some $a\in\Sigma_1$ such that $s_0^1\xrightarrow{a_0\cdot\ldots\cdot a_{k-1}\cdot a| \Sigma_1}_{O_1} s$ and $H=\mathit{after}_1(s)$} \\[-8mm]
\end{align*}
Informally, the first rule enforces that every call/return action $a$ is followed by a program action $\record{a}$.
The second rule ensures that $S_1$ is permissive enough, i.e., it allows all the successors of the current object state that have different $\mathit{after}_1$ images (for a sequence $\sigma$, $\sigma[\cdot a]$ denotes a sequence where the character $a$ is optional). The third rule ensures that every $\guess{H}$ is followed by an action leading to an object state $s$ with $H=\mathit{after}_1(s)$. Collectively, these last two cases ensure that every action $a$ of $O_1$ is  preceded by a $\guess{H}$ program action where $H$ is the set of call/return sequences admitted from the post-state of $a$.

Although $S_1$ does not admit all the executions of $O_1$ (because of the arbitrary choice of $a$ in the third case above), we show that the set of executions it admits simulate all the executions of $O_1$: let $O_1[S_1]$ be an LTS representing the set of executions of $O_1$ consistent with $S_1$ (obtained from the set of executions of $P$ consistent with $S_1$ by projecting out the program state and actions). We show that the relation $F_1$ between states of $O_1$ and $O_1[S_1]$, respectively, defined
by $(s,s')\in F_1$ iff $\mathit{after}_1(s)=\mathit{after}_1(s')$,
is a $(C\cup R)$-forward simulation from $O_1$ to $O_1[S_1]$. The fact that it relates the initial object states $s_0^1$ and $s_0^1$ is trivial. Now, let $s,s_1\in Q_1$ and $a\in \Sigma_1$ such that $(s,a,s_1) \in \delta_1$ and $(s,s')\in F_1$. Using a simple induction on the length of executions, it can be shown that there exists a state $s_1'$ with $\mathit{after}_1(s_1)=\mathit{after}_1(s_1')$ such that $(s',b,s_1')$ for some action $b$. If $a\in C\cup R$, then $b=a$ because otherwise, the continuations with call/return actions admitted from $s_1$ will be different from those admitted from $s_1'$ (for instance, if $a$ is a call action and $b$ is an internal action, then the matching return action will be eventually enabled in executions starting from $s_1$ but not from $s_1'$, at least not before $a$ occurs). For the same reason, if $a$ is an internal action, then $b$ is also an internal action. This concludes the proof that $F_1$ is a forward simulation.

Since $O_1\leq_s O_2$, there exists a scheduler $S_2$ such that $T(P \times O_1,S_1)|\Sigma_p=T(P \times O_2,S_2)|\Sigma_p$.
Let $P[O_2,S_2]$ denote the LTS representation of the set of executions of $P$ with $O_2$ and consistent with $S_2$ (explained in Section~\ref{sec:sor}). It can be easily seen that $P[O_2,S_2]$ is $\Sigma_p$-deterministic (the interleaving of a sequence of $\Sigma_p$ actions with internal actions of $O_2$ is uniquely determined by $S_2$ because it is a deterministic scheduler). Since $T(P[O_1,S_1])|\Sigma_p\subseteq T(P[O_2,S_2])|\Sigma_p$,\footnote{Note that $T(P \times O_i,S_i)$ and $T(P[O_i,S_i])$ with $i\in \{1,2\}$ denote exactly the same set of traces.} we get that there exists a $\Sigma_p$-forward simulation $F_{S_1,S_2}$ from  $P[O_1,S_1]$ to $P[O_2,S_2]$. Such a forward simulation defines a relation between states of $O_1$ and $O_2$, respectively, by removing the program state, i.e., $s_1$ and $s_2$ are related whenever $((s_0,s_1),(s_0,s_2))\in F_{S_1,S_2}$. For simplicity, this relation is denoted by $F_{S_1,S_2}$ as well. Because of the $\record{a}$ actions in $\Sigma_p$, we get that $F_{S_1,S_2}$ is a $(C\cup R)$-forward simulation from $O_1[S_1]$ to $O_2$. It is easy to check that $F_1\circ F_{S_1,S_2}$ (where $\circ$ is the usual composition of relations) is a $(C\cup R)$-forward simulation from $O_1$ to $O_2$.
\vspace{-2.5mm}
\end{proof}

The two lemmas above imply that:

\vspace{-2mm}
\begin{theorem}\label{th:strong-ref}
$O_1\leq_s O_2$ iff there exists a $(C\cup R)$-forward simulation from $O_1$ to $O_2$.
\vspace{-1.5mm}
\end{theorem}

The fact that forward simulations are \emph{necessary} for strong refinement allows to derive in a simple way compositional methods for proving strong refinement.
In the following we consider the case of \emph{composed} objects defined as a product of a fixed set of objects, and \emph{parametrized} objects defined from a set of ``base'' objects which are considered as parameters.

We show that strong refinement is a \emph{local} property, i.e., it holds for composed objects if and only if it holds for individual objects in this composition. As usual, we consider compositions of objects with disjoint states and sets of actions. Indeed, any forward simulation between composed objects can be ``projected'' to a set of forward simulations that hold between individual objects, and vice versa. We state this result for compositions of two objects, the extension to an arbitrary number of objects is obvious.

\vspace{-1.5mm}
\begin{theorem}
Let $O_1$ and $O_2$, resp., $O_1'$ and $O_2'$, be two objects over an alphabet $\Sigma$, resp., $\Sigma'$, such that $\Sigma\cap \Sigma'=\emptyset$. Then, $O_1\times O_1'\leq_s O_2\times O_2'$ iff $O_1\leq_s O_2$ and $O_1'\leq_s O_2'$.
\vspace{-1.5mm}
\end{theorem}

Next, we consider the case of parametrized objects whose implementation is parametrized by a set of base objects, e.g., snapshot objects defined from a set of atomic registers. We show that if the parametrized object is a strong refinement of an abstract specification $S$ assuming that the base objects behave according to their own abstract specifications $S_i$, then instantiating any base object with an implementation that is a strong refinement of $S_i$ leads to an object which remains a strong refinement of $S$. Assuming for simplicity only one base object, a parametrized object $O$ can be formally defined as a product $O=S_1\times C$ where $S_1$ is the base object's specification and $C$ is the context in which this object is used to derive the implementation of $O$\,\footnote{For a parametrized object $O=S_1\times C$, the alphabets of $S_1$ and $C$ share the call/return actions of $S_1$ (the base object) and the alphabet of $C$ contains the call/return actions of $O$. This is different from the the composition of two objects $O_1\times O_1'$ where the alphabets of $O_1$ and $O_1'$ are disjoint.}. To distinguish parametrization from composition, we use $O[S_1]$ to denote an object parametrized by a base object $S_1$. The next result is an immediate consequence of the fact that the forward simulation admitted by the base object can be composed\,\footnote{Here, we refer to classical composition of relations.} with the one admitted by the parametrized object (assuming base object's specification) to derive a forward simulation for the instantiation.

\vspace{-1.5mm}
\begin{theorem}
If $O[S_1]\leq_s S$ and $O_1\leq_s S_1$, then $O[O_1]\leq_s S$.
\vspace{-1.5mm}
\end{theorem}

Finally, it can be shown that the existence of forward simulations is equivalent to \emph{strong linearizability}~\cite{DBLP:conf/stoc/GolabHW11} when concrete objects are related to \emph{atomic} abstract objects. Thus, let $O_2$ be an atomic object defined by a set of sequential histories $\mathit{Seq}$, i.e., $H(O_2)=\{h: \exists h'\in\mathit{Seq}.\ h\sqsubseteq h'\}$. We say that an object $O_1$ is \emph{strongly linearizable} w.r.t. $O_2$, written $O_1\sqsubseteq_s O_2$, when there exists a function $f:T(O_1)\rightarrow Seq$ such that (1) for any trace $\tau\in T(O_1)$, $\mathit{hist}(\tau)\sqsubseteq f(\tau)$, and (2) $f$ is prefix-preserving, i.e., for any two traces $\tau_1,\tau_2\in T(O_1)$ such that $\tau_1$ is a prefix of $\tau_2$, $f(\tau_1)$ is a prefix of $f(\tau_2)$. It can be shown that the function $f$ induces a forward simulation and vice-versa (the proof is given in Appendix~\ref{app:strong-lin}).

\vspace{-1.5mm}
\begin{theorem}\label{th:strong-lin}
If $O_2$ is atomic, then $O_1\sqsubseteq_s O_2$ iff there exists a $(C\cup R)$-forward simulation from $O_1$ to $O_2$.
\vspace{-2.5mm}
\end{theorem}

 \section{Strong Observational Refinements of Non-Atomic Specifications}
\label{sec:applications}

\vspace{-1mm}
We demonstrate that many concurrent objects defined in the literature are strong observational refinements of much simpler abstract objects, even though not necessarily atomic. We focus on objects which are not strongly linearizable, since by Theorem~\ref{th:strong-lin}, the latter are strong refinements of atomic objects.

\lstset{numbers=left,  stepnumber=1, numberblanklines=false, numberstyle=\tiny,basicstyle=\scriptsize, numbersep=3pt, escapeinside={/*}{*/}}
\begin{figure}[t]
\center
  \footnotesize
  \begin{minipage}[c]{75mm}
    \begin{program}
procedure update(i,data)
  mem[i] = data;

procedure scan()
  for i = 1 to n do r1[i] = mem[i];
  repeat
    r2 = r1;
    for i = 1 to n do r1[i] = mem[i];
  until r1 == r2
  return r1;
    \end{program}
  \end{minipage}
  \begin{minipage}[c]{50mm}
    \begin{program}
procedure update(i,data)
  mem[i] = data;

procedure scan()
  while ( nondet )
    r = atomic_snapshot();
    snaps = snaps $\cdot$ r;
  return r1 $\in$ snaps;
    \end{program}
  \end{minipage}
  \vspace{-4mm}
  \footnotesize \caption{ A snapshot object (on the left), and a concurrent specification (on the right). The shared state of both is an array {\tt mem} of size {\tt n}. The local variables {\tt r1}, {\tt r2}, and {\tt r} are arrays of size {\tt n} (initialized to the same value as {\tt mem}). The local variable {\tt snaps} is a sequence of arrays of size {\tt n} ($\cdot$ denotes the concatenation operator), initially containing a single array which equals the initial value of {\tt mem}. The use of {\tt nondet} means that the loop is executed for an arbitrary number of times. The procedure {\tt atomic\_snapshot} returns a snapshot of {\tt mem} in a single step executed in isolation.}
  \label{fig:snapshots}
  \vspace{-3mm}
\end{figure}

Figure~\ref{fig:snapshots} lists an implementation of a snapshot object with two methods {\tt update(i,data)} for writing the value {\tt data} to a location {\tt i} of a shared array {\tt mem},  and {\tt scan()} for returning a snapshot of the array {\tt mem}.\footnote{This is a simplified version of the snapshot object defined by Afek et al.~\cite{DBLP:journals/jacm/AfekADGMS93}.} While the implementation of {\tt update} is obvious, a {\tt scan} operation performs several ``collect'' phases, where it reads successively all the cells of {\tt mem}, until two consecutive phases return the same array.

This object does \emph{not} admit a forward simulation towards the standard atomic specification where the method {\tt scan} takes a \emph{single} instantaneous snapshot of the entire array which is subsequently returned (it is not a strong refinement of such a specification). Intuitively, this holds because the linearization point of  {\tt scan} depends on future steps in the execution, e.g., a read in the second {\tt for} loop is a linearization point only if it is not followed by updates on array cells before and after the current loop index. This is exactly the scenario in which backward simulations are necessary, intuitively, reading an execution backwards it is possible to identify precisely the linearization points of {\tt scan} invocations. The impossibility of defining such a forward simulation is also a consequence of the fact that this object is not strongly linearizable~\cite{DBLP:conf/stoc/GolabHW11}.

However, this object is a strong refinement of the simpler ``concurrent'' specification given on the right of Figure~\ref{fig:snapshots} (see Appendix~\ref{app:snapshot}). The implementation of {\tt update} remains the same, while a {\tt scan} operation performs a sequence of \emph{instantaneous} snapshots of the entire array {\tt mem} and returns \emph{any} snapshot in this sequence. Compared to the implementation on the left, it is simpler because it does not allow that reading the array {\tt mem} is interleaved with other operations. However, it is not atomic since an execution of {\tt scan} contains more than one step. In comparison with the atomic specification, the sequence of snapshots in {\tt scan} allows that an adversary (scheduler) decides on the return value ``lazily'' after observing other invocations, e.g., updates, exactly as in the concrete implementation. Therefore, the abstract specification in Figure~\ref{fig:snapshots} can be used while reasoning about hyperproperties of clients, which is not the case for the atomic specification.

Beyond snapshot objects, Bouajjani et al.~\cite{DBLP:conf/cav/BouajjaniEEM17} show that a similar simplification holds even for concurrent queues and stacks which are not strongly linearizable, e.g., Herlihy\&Wing queue~\cite{DBLP:journals/toplas/HerlihyW90} and Time-Stamped Stack~\cite{DBLP:conf/popl/DoddsHK15}. These objects admit forward simulations towards ``concurrent'' specifications where roughly, the elements are stored in a partially-ordered set instead of a sequence (which is consistent with the real-time order between the enqueues/pushes that added those elements). The enqueues/pushes have no internal steps, while the dequeues/pops have a single internal step which roughly, corresponds to a linearization point that extracts a minimal (for queues) or maximal (for stacks) element from the partially-ordered set. The stack of Afek at al.~\cite{DBLP:journals/dc/AfekGM07} can also be proved to be a strong refinement of such a specification.
These forward simulations imply that these objects are strong refinements of their specifications.

 \vspace{-1mm}
\section{Related Work and Discussion}
\label{sec:summary}

\vspace{-1mm}
An important contribution of our paper is to put the work on strong linearizability~\cite{DBLP:conf/wdag/DenysyukW15,DBLP:conf/podc/HelmiHW12,DBLP:conf/stoc/GolabHW11}
in the context of standard results concerning
hyperproperties~\cite{DBLP:journals/jcs/ClarksonS10,DBLP:conf/post/ClarksonFKMRS14}
and property-preserving refinements~\cite{DBLP:conf/icalp/AlurCZ06,DBLP:journals/iandc/LynchV95,DBLP:conf/sp/McLean94}.
McLean~\cite{DBLP:conf/sp/McLean94} showed that refinements do not
preserve security properties, which were later found to be instances
of the more generic notion of hyperproperty~\cite{DBLP:journals/jcs/ClarksonS10}.
By exploiting the equivalence between linearizability and refinement~\cite{DBLP:conf/popl/BouajjaniEEH15,DBLP:journals/tcs/FilipovicORY10},
our paper clarifies that a stronger notion of linearizability is needed
because standard linearizability does not preserve hyperproperties.

Our notion of strong observational refinement is a variation of the
hyperproperty-preserving refinement introduced in~\cite{DBLP:journals/jcs/ClarksonS10},
which takes into account the specificities of concurrent object clients.
The relationship between forward simulations and preservation of
hyperproperties has been investigated in~\cite{DBLP:conf/icalp/AlurCZ06}.
They show that the existence of forward simulations is \emph{sufficient}
for preserving some specific class of hyperproperties (information-flow
security properties like non-interference),
corresponding to the straightforward direction of Theorem~\ref{th:strong-ref}
(Lemma~\ref{lem:fsim1});
they also show that their condition is \emph{not necessary} in their context.
In contrast, our work shows that the existence of forward simulations 
is \emph{both necessary and sufficient} for preserving any hyperproperty
in the context of concurrent object clients.

An important consequence of our results is that strong 
linearizability~\cite{DBLP:conf/stoc/GolabHW11}
is equivalent to the existence of a forward simulation towards an atomic specification.
The equivalence to the well-studied notion of forward simulation immediately
implies methods for composing concurrent objects,
in particular, \emph{locality} and \emph{instantiation}.
This stands in contrast to the effort needed to prove similar results
in~\cite{DBLP:conf/stoc/GolabHW11} and~\cite{OvensWoelfel2019}.

While \cite{DBLP:conf/stoc/GolabHW11} relates strong linearizability to
an ad-hoc notion of reasoning about randomized programs when
replacing objects by their atomic specifications,
the equivalence we prove implies that strong linearizability
is necessary and sufficient for preserving hyperproperties in this context.
Note that forward simulations are more general than strong linearizability.
Section~\ref{sec:applications} presents several objects which are \emph{not} strongly linearizable, but which admit forward simulations towards non-atomic abstract specifications. 
Our results imply that it is sound to use such specifications 
when reasoning about hyperproperties of client programs. 
Moreover, as opposed to strong linearizability, 
forward simulations are applicable to \emph{interval-linearizable} 
objects~\cite{DBLP:journals/jacm/CastanedaRR18}, 
which do not have any atomic specification,
but are essentially LTSs as in our formalization.

Finally, Bouajjani et al.~\cite{DBLP:conf/cav/BouajjaniEEM17} 
show that intricate implementations of concurrent stacks and 
queues like Herlihy\&Wing queue~\cite{DBLP:journals/toplas/HerlihyW90} 
and Time-Stamped Stack~\cite{DBLP:conf/popl/DoddsHK15} 
admit forward simulations towards non-atomic abstract specifications, 
but they do not discuss the connection between existence of 
forward simulations and preservation of hyperproperties,
which is the main contribution of our paper.

Our definition of strong observational refinement and its deep relation
to forward simulations deepens our understanding of the role of strong
linearizability in preserving hyperproperties.
We plan to explore strong observational refinement of almost-atomic
objects and develop additional proof methodologies.
Also, our notion of strong observational refinement uses deterministic
schedulers that model strong adversaries
w.r.t.~Aspnes' classification~\cite{DBLP:journals/dc/Aspnes03},
and it is interesting to explore variations of this notion
that take into account other adversary models.

\clearpage

\bibliographystyle{acm}
\bibliography{misc,dblp}

\appendix
\section{Proof of Theorem~\ref{th:sor_preservation}}

\begin{proof}
Assume that $O_1\leq_s O_2$ and $P\times O_2\models\varphi$ for some hyperproperty $\varphi$ of $P$. Let $S_1$ be a deterministic scheduler admitted by $P \times O_1$. Since $O_1\leq_s O_2$, there exists a deterministic scheduler $S_2$ admitted by $P \times O_2$ such that $T(P \times O_1,S_1)|\Sigma_p=T(P \times O_2,S_2)|\Sigma_p$. Since, $P\times O_2\models\varphi$, we get that $T(P \times O_2,S_2)|\Sigma_p\in \varphi$, which implies that $T(P \times O_1,S_1)|\Sigma_p\in \varphi$. Therefore, $P\times O_1\models\varphi$.
\end{proof}

 \section{Proof of Lemma~\ref{lem:fsim1}}\label{app:suff}

\begin{proof}
Let $O_1= (Q_1,\Sigma_1,  s_0^1, \delta_1)$ and $O_2=(Q_2,\Sigma_2,  s_0^2, \delta_2)$ be two objects, and
$F$ a $(C\cup R)$-forward simulation from $O_1$ to $O_2$. Let $P$ be a program and $S_1$ a deterministic scheduler admitted by $P\times O_1$.
We define a deterministic scheduler $S_2$ admitted by $P\times O_2$ inductively as follows:
$$
S_2(\epsilon)  = S_1(\epsilon) \hspace{1cm}
S_2(\overline{\tau})  = \left\{
                \begin{array}{ll}
                  S_1(\tau),  \mbox{if $S_1(\tau)\subseteq \Sigma_p$} \\[1mm]
                  F(S_1(\tau)), \mbox{otherwise}
                \end{array}
              \right.
$$
where $\overline{\tau}$ is the trace of $P\times O_2$ associated to the trace $\tau$ of $P\times O_1$ by previous iterations of this inductive definition, and $F(S_1(\tau))$ is a sequence of actions of  $O_2$ simulating the action $S_1(\tau)$ (which is an action of $O_1$ in this case) in the state reached after the trace $\tau|\Sigma_1$. Formally, if $s_0^1\xrightarrow{\tau|\Sigma_1}_{O_1}s_1$, then a simple induction on the length of executions can show that $(s_1,s_2)\in F$ where $s_0^2\xrightarrow{\overline{\tau}|\Sigma_2}_{O_2}s_2$.
Then, since $s_1\xrightarrow{S_1(\tau)}_{O_1} s_1'$ is a transition of $O_1$ and $F$ is a forward simulation, we get that there exists $s_2'$ such that $(s_1',s_2')\in F$ and $s_2\xrightarrow{\sigma}_{O_2} s_2'$ and $\sigma|(C\cup R)=S_1(\tau)|(C\cup R)$. We define $F(S_1(\tau))=\sigma$.

The scheduler $S_2$ is a slight deviation from the definition of a deterministic scheduler because $F(S_1(\tau))$ is not necessarily a singleton. However, the definition of $S_2$ can be adapted easily such that this sequence of steps is performed one by one.

Since $F$ is a $(C\cup R)$-forward simulation, $T(P \times O_1,S_1)|\Sigma_p\subseteq T(P \times O_2,S_2)|\Sigma_p$ is obvious. The reverse follows from the fact that $S_2$ is defined inductively following the definition of $S_1$.
\end{proof}
 \section{Relating Forward Simulations and Strong Linearizability}\label{app:strong-lin}

The following result uses a concrete definition of atomic specifications as LTSs. Thus, an \emph{atomic} object is an LTS $O=(Q,\Sigma,s_0,\delta)$ where the states are pairs formed of a history $h$ and a linearization $h_s$ of $h$, i.e., $Q=\{(h,h_s): h_s\in Seq\mbox{ and }h\sqsubseteq h_s\}$, the internal actions are \emph{linearization point} actions $\mathit{lin}(k)$ (for linearizing an operation with identifier $k$), i.e., $\Sigma=C\cup R\cup \{\mathit{lin}(k): k\in \<Ops>\}$, the initial state contains an empty history and linearization, i.e., $s_0=(\epsilon,\epsilon)$, and the transition relation is defined by: $((h,h_s),a,(h',h_s'))\in \delta$ if
\begin{align*}
& a\in C \implies h'=h\cdot a \mbox{ and } h_s'=h_s   \\
& a\in R\implies h'=h\cdot a\mbox{ and $h_s'=h_s$ and $a$ occurs in $h_s'$} \\
& a=lin(k) \implies h'=h \mbox{ and $h_s'=h_s\cdot call(m,d_1,k)\cdot ret(m,d_2,k)$, for some $m$, $d_1$, and $d_2$.}
\end{align*}
Call actions are only appended to the history $h$, return actions ensure that additionally, the linearization $h_s'$ contains the corresponding operation, and linearization points extend the linearization with a new operation.

\begin{lemma}\label{lem:strong-lin1}
Let $O_1$ be an object and $O_2$ an atomic object. If $O_1$ is strongly linearizable w.r.t. $O_2$, then there exists a $(C\cup R)$-forward simulation from $O_1$ to $O_2$.
\end{lemma}
\begin{proof}
Let $\mathit{state}(\tau)$ denote the state of $O_1$ reached after a trace $\tau$ (since $O_1$ is deterministic, this state is unique). Also, let $F$ be a relation between states of $O_1$ and $O_2$ defined by $(s_1,s_2)\in F$ iff there exists a trace $\tau$ such that $s_1=\mathit{state}(\tau)$ and $s_2=(\mathit{hist}(\tau),f(\tau))$ (by the definition of $f$ in strong linearizability, the latter is a valid state of $O_2$)~\footnote{We use the definition of atomic objects given in Section~\ref{sec:objects}.}.

We show that $F$ is a $(C\cup R)$-forward simulation from $O_1$ to $O_2$. The fact that it relates the initial object state $s_0^1$ and the initial state $(\epsilon,\epsilon)$ of $O_2$ is trivial. Now, let
$s_1,s_1'\in Q_1$, $a\in \Sigma_1$, and $s_2\in Q_2$, such that $(s_1,a,s_1') \in \delta_1$ and $(s_1,s_2)\in F$. We have to show that there exists $s_2'\in Q_2$ and $\sigma\in \Sigma_2$ such that $(s_1',s_2')\in F$, $s_2 \xrightarrow{\sigma}_{O_2} s_2'$, and $\sigma| (C\cup R)=a| (C\cup R)$.
(
Let $\tau$ be a trace such that $s_1=\mathit{state}(\tau)$ and $s_2=(\mathit{hist}(\tau),f(\tau))$. Then, $s_1'=\mathit{state}(\tau\cdot a)$ and $f(\tau)$ is a prefix of $f(\tau\cdot a)$. Several cases are to be discussed:
\begin{itemize}
	\item if $a\in C$, then $\mathit{hist}(\tau\cdot a)=\mathit{hist}(\tau)\cdot a$ and $\mathit{hist}(\tau)\cdot a\sqsubseteq f(\tau)$ provided that $\mathit{hist}(\tau)\sqsubseteq f(\tau)$. Therefore, $s_2\xrightarrow{a}_{O_2}(\mathit{hist}(\tau)\cdot a,f(\tau))$ and $(s_1',(\mathit{hist}(\tau)\cdot a,f(\tau)))\in F$.
	\item if $a\in R$ and the operation identifier $k$ in $a$ occurs in $f(\tau)$, then $s_2\xrightarrow{a}_{O_2}(\mathit{hist}(\tau)\cdot a,f(\tau))$ and $(s_1',(\mathit{hist}(\tau)\cdot a,f(\tau)))\in F$ like above. If $k$ does not occur in $f(\tau)$, then $f(\tau\cdot a)=f(\tau)\cdot c\cdot a$ where $c$ is the call action corresponding to $a$ (otherwise, $f(\tau\cdot a)$ would not be a linearization of $\mathit{hist}(\tau\cdot a)$). By the definition of $O_2$, we have that $s_2\xrightarrow{\mathit{lin}(k)\cdot a}_{O_2}(\mathit{hist}(\tau)\cdot a,f(\tau\cdot a))$ which concludes the proof of this case.
	\item if $a\not\in C\cup R$, then $f(\tau\cdot a)$ is obtained from $f(\tau)$ by appending some sequence of operations with identifiers $k_1$, $\ldots$, $k_n$ (this follows from the fact that $f$ is prefix-preserving). Then, $s_2\xrightarrow{\mathit{lin}(k_1)\cdot\ldots\cdot \mathit{lin}(k_n)}_{O_2}(\mathit{hist}(\tau),f(\tau\cdot a))$ and $(s_1',(\mathit{hist}(\tau),f(\tau\cdot a)))\in F$ (because in this case, $\mathit{hist}(\tau\cdot a)=\mathit{hist}(\tau)$).
\end{itemize}
\end{proof}

\begin{lemma}
Let $O_1$ be an object and $O_2$ an atomic object. If there exists a $(C\cup R)$-forward simulation from $O_1$ to $O_2$, then $O_1$ is strongly linearizable w.r.t. $O_2$.
\end{lemma}
\begin{proof}
Let $F$ be a $(C\cup R)$-forward simulation from $O_1$ to $O_2$. We define a function $f:T(O_1)\rightarrow Seq$ by $f(\tau)=h_s$ where $h_s$ satisfies $(\mathit{state}(\tau),(h,h_s))\in F$. The fact that $\mathit{hist}(\tau)\sqsubseteq f(\tau)$ for every trace $\tau$ follows from the definition since $(h,h_s)$ is a valid state of $O_2$ and $h=\mathit{hist}(\tau)$ (because $F$ preserves call and return actions). The fact that $f$ is prefix-preserving follows from the fact that $F$ is a forward simulation.
\end{proof}  \section{Snapshot Object}\label{app:snapshot}

We show that the ``concrete'' implementation of the snapshot is a strong refinement of this specification using our characterization in terms of forward simulations.
The pseudo-code descriptions in Figure~\ref{fig:snapshots} define LTSs whose states are formed of an array {\tt mem} and for each active method invocation, a valuation of its local variables including the current program location. The local variable valuation of an invocation $k$ in a state $s$ is denoted by $s[k]$. We define a forward simulation $F$ from the LTS of the ``concrete'' snapshot object to the LTS of its specification. Thus, $(s,s')\in F$ iff the {\tt mem} arrays in $s$ and $s'$ are the same, the two states contain the same set of active method invocations (matched using their identifiers), and for each {\tt scan} invocation $k$, the local state $s[k]$ in the concrete object, i.e., a valuation of {\tt r1} and {\tt r2}, and the current program location {\tt pc}, is related to the local state $s'[k]$ in the specification, i.e., a valuation of {\tt snaps}, if the following holds:
\begin{align}
\hspace{-3mm}
\bigwedge_{{\texttt r}\in \{\texttt{r1},\texttt{r2}\}}{\sf valid}(\text{\tt r},\text{\tt pc})\land {\sf fst}(\text{\tt r2},\text{\tt n}) \leq {\sf fst}(\text{\tt r1},0)\land {\sf last}(\text{\tt snaps})=\text{\tt mem} \land (\texttt{pc}=10 \implies \texttt{r1}\in \texttt{snaps})\label{eq:sim}
\end{align}
The predicate ${\sf valid}$ says that {\tt r1} and {\tt r2} are obtained by reading from a sequence of snapshots in {\tt snaps} such that the value at position $j>i$ comes from a snapshot which is at least as recent as the one from which the value at position $i$ is taken (with a slight deviation for {\tt r1}). Thus, let ${\sf fst}(v,i)$ be the smallest index of a snapshot $s\in\texttt{snaps}$ which contains the value $v$ on the $i$-th position. We define ${\sf fst}(\text{\tt r1},i)$ as $\infty$ when the invocation $k$ is ``inside'' a {\tt for} loop and not yet set the $i$-th position of {\tt r1} (which can be derived from the current program location {\tt pc}), or ${\sf fst}(\texttt{r1}[i],i)$, otherwise. Also, ${\sf fst}(\text{\tt r2},i)$ is simply ${\sf fst}(\texttt{r2}[i],i)$ for every $i$. Then,
$
{\sf valid}(\text{\tt r},\text{\tt pc}) ::= \forall i,j.\ i < j \implies {\sf fst}(\text{\tt r},i) \leq {\sf fst}(\text{\tt r},j).
$ 
Besides the ${\sf valid}$ predicates, (\ref{eq:sim}) states that {\tt r1} reads from more recent snapshots than {\tt r2}, that the last element of {\tt snaps} is the current value of {\tt mem}, and that {\tt r1} is a member of {\tt snaps} when the invocation reaches line 10 (right after the equality test).

We show that indeed the specification can mimic every step of the concrete implementation w.r.t. the simulation relation $F$ described above. The most interesting part concerns internal steps (every call/return action of the concrete implementation is mimicked by exactly the same action in the specification).

The internal step of {\tt update} is simulated by the same step of {\tt update} in the specification followed by a sequence of steps in which every {\tt scan} takes an instantaneous snapshot of {\tt mem} and appends it to its own {\tt snaps} variable. This ensures the last element of every {\tt snaps} variable is the value of {\tt mem} after the update.

The internal steps of {\tt scan} are all simulated by stuttering steps in the specification, i.e., if  $(s_1,s_2)\in F$ and $s_1'$ is a successor of $s_1$ by such a step, then $(s_1',s_2)\in F$.
Thus, for a step of {\tt scan} reading a new position in {\tt mem}, i.e., {\tt r1[i] = mem[i]}, we have that $(s_1',s_2)\in F$ because the value of {\tt r1[i]} is taken directly from {\tt mem} which is the last snapshot in every {\tt snaps} variable (therefore, the requirements on {\tt r1} from (\ref{eq:sim}) hold when adding this new value).
For a step of {\tt scan} where the equality between {\tt r1} and {\tt r2} is established, we have that $(s_1',s_2)\in F$ because this equality implies $\texttt{r1}\in \texttt{snaps}$. Indeed, the only way that {\tt r1} which reads from more recent snapshots in {\tt snaps} than {\tt r2} is equal to the latter is that they both read from the \emph{same} snapshot in {\tt snaps}. The simulation of the other internal steps of {\tt scan} is obvious.

\end{document}